\documentclass[a4paper]{article}
\usepackage{amssymb}
\usepackage{amsthm}
\usepackage{amsmath}
\usepackage{graphicx}
\usepackage{color}
\usepackage{dsfont}
\usepackage{bm}
\usepackage{pgf,tikz}
\usepackage{hyperref}
\usepackage{authblk}
\usetikzlibrary{arrows}

\newtheorem{theorem}{Theorem}

\theoremstyle{definition}

\newtheorem{remark}[theorem]{Remark}

\newcommand{\R}{\mathbb{R}}

\newcommand{\E}{\mathbb{E}}

\newcommand{\Filt}{\mathcal{F}}

\newcommand{\Prob}{\mathbb{P}}

\newcommand{\STOP}{\mathcal{T}}

\newcommand{\argmin}{\operatornamewithlimits{arg\,min}}

\newcommand{\esssup}{\operatornamewithlimits{ess\,sup}}

\let\inf\relax \DeclareMathOperator*\inf{\vphantom{p}inf}

\newcommand\I{\mathds{1}}

\title{Pricing Bounds for VIX Derivatives via Least Squares Monte Carlo}
\author{Ivan Guo \qquad Gregoire Loeper}
\affil{\normalsize School of Mathematical Sciences\\ Clayton Campus, Monash University, VIC, 3800, Australia}

\begin{document}
\maketitle

\begin{abstract}
Derivatives on the Chicago Board Options Exchange volatility index (VIX) have gained significant popularity over the last decade. The pricing of VIX derivatives involves evaluating the square root of the expected realised variance which cannot be computed by direct Monte Carlo methods. Least squares Monte Carlo methods can be used but the sign of the error is difficult to determine. In this paper, we propose new model independent upper and lower pricing bounds for VIX derivatives. In particular, we first present a general stochastic duality result on payoffs involving concave functions. This is then applied to VIX derivatives along with minor adjustments to handle issues caused by the square root function. The upper bound involves the evaluation of a variance swap, while the lower bound involves estimating a martingale increment corresponding to its hedging portfolio. Both can be achieved simultaneously using a single linear least square regression. Numerical results show that the method works very well for VIX futures, calls and puts under a wide range of parameter choices.


\end{abstract}

\section{Introduction}
The Chicago Board Options Exchange volatility index, commonly known as VIX, measures the volatility of the S\&P500 index. Formally, the VIX is the square root of the expected integrated variance (often called the realised variance) over a 30 day period, multiplied by an annualisation factor. In practice, it is calculated using a weighted sum of options on the S\&P500 index and it coincides with the square root of the par variance swap rate. The VIX itself is not a tradable asset, but VIX derivatives such as futures and options are. VIX futures began trading in 2004 while VIX options began in 2006. Since then, VIX derivatives have gained significant popularity as they allow traders to gain direct exposure to the volatility of the S\&P500 index without having to hold options the index.

In literature, there have been many theoretical approaches to the pricing of VIX derivatives. 
In earlier works, the authors focussed on finding analytical pricing formulae for volatility derivatives under particular volatility dynamics. Some examples include Whaley \cite{Whaley} (geometric Brownian motion), Gr\"unbichler and Longstaff \cite{Grunbichler} (square root process), Detemple and Osakwe \cite{Detemple} (log-normal Ornstein-Ulenbeck process). By only considering volatility futures and vanilla options as opposed to VIX derivatives, these works do not explicitly deal with the integrated variance term. This is rectified by Zhang and Zhu \cite{Zhang} who derived an analytical formula for the price of VIX futures under the Heston model. Furthermore they supplemented their work with empirical analyses by calibrating the model against historical VIX data. This pricing result was further generalised by Lian and Zhu \cite{Lian} to the Heston model with jumps via a characteristic function approach. Further progress was made for cases where the variance process follows a square root process with jumps (Sepp \cite{Sepp}) and a 3/2 process with jumps (Baldeaux and Badran \cite{Baldeaux}). Finally, some author undertook an alternative approach which directly models the variance swaps instead of the volatility. This allows for the consistent modelling and the simultaneous calibration of both index options and VIX derivatives. See Cont and Kokholm \cite{Cont} for an example of this approach.

In terms of numerical methods, PDE methods work well but only if the underlying dynamic is Markovian and resides in a low dimensional space. Due to the non-linearity of the square root function in the definition of the VIX, the price of VIX futures is highly model-dependent and cannot be inferred from direct Monte Carlo simulations. Instead, the evaluation of the conditional expectation of the integrated variance can be handled by nested simulations or least squares regressions. Nested Monte Carlo has good accuracy, but it is computationally expensive. Least square Monte Carlo approaches, popularised by Longstaff and Schwartz \cite{Longstaff} for Bermudan options, are much faster. Although the results are asymptotically unbiased, it is usually difficult to determine the sign of the error, which can be a useful piece of information in risk management. Rogers \cite{Rogers} as well as Haugh and Kogan \cite{Haugh} proposed a stochastic duality result which produces an upper bound to Bermudan option prices, complementing the original least squares Monte Carlo method which naturally provides a lower bound via suboptimal exercise policies. The quality of the upper bound relies on the identification of a martingale which majorises the price process. Andersen and Broadie \cite{Andersen} suggested to estimate the martingale using nested Monte Carlo. Later on more efficient approaches were found in various works such as Schoenmakers et al.\! \cite{Schoenmakers}. An overview of these upper bound methods without using nested simulations can be found in Joshi and Tang \cite{Joshi}.

In this paper, we present a new application of the stochastic duality and the least squares Monte Carlo methods to VIX derivatives, resulting in true upper and lower pricing bounds.
Although, at a first glance, the stochastic duality approach is not applicable to derivatives such as the VIX future due to the lack of early exercise features, we show that VIX derivatives can in fact be placed under the same framework using Legendre transforms which converts the VIX derivatives to a variant of the chooser option (see Remark \ref{remchooser}). Then by using techniques similar to Schoenmakers et al.\! \cite{Schoenmakers}, we perform a single least squares Monte Carlo to compute the required conditional expectation and martingale increment, which are used to evaluate the pricing bounds. The main results of the paper are Theorems \ref{propav01} and \ref{thmbb01}. Theorem \ref{propav01} presents a general stochastic duality result on payoffs involving concave functions. Theorem \ref{thmbb01} applies it to VIX derivatives, with minor adjustments to handle issues caused by the square root functions. Despite focussing our presentation on VIX derivatives in the local-stochastic volatility model, the techniques and results described in this paper are in fact completely model independent and directly applicable to many other derivatives in various settings. 

The paper is organised as follows. Section \ref{sec2} introduces the underlying framework and defines the VIX as well as its derivatives. Then in Section \ref{sec3}, theoretical upper and lower bounds are derived, along with techniques to handle the square root function in VIX. Section \ref{sec4} describes the Monte Carlo algorithm in detail while Section \ref{sec5} provides some numerical examples. Finally, Section \ref{sec6} contains some concluding remarks.

\section{Framework}\label{sec2}
The core techniques and results of this paper are completely model independent, but for the sake of presentation and readability, we have chosen to focus on the following model as an example. Let $(\Omega, \Filt, \Prob)$ be a filtered probability space where the filtration $\Filt$ represents the information flow available to market participants and $\Prob$ is a pricing measure. Consider the following general local-stochastic volatility (LSV) model for the price of a stock or a stock index $S_t$,
\begin{align*}
dS_t&= \mu(t,S_t) S_t dt+  \sigma(t,S_t, V_t)S_t dW^S_t,\\
dV_t&=a(t,V_t) dt+ b(t,V_t) dW^V_t,\\
\langle dW^S_t,dW^V_t\rangle & = \rho(t,S_t,V_t) dt.
\end{align*}
where $W^S_t$ and $W^V_t$ is are standard Brownian motions. For simplicity, the interest rate is set to be zero. Before continuing, let us again emphasise that the main results of the paper, Theorems \ref{propav01} and \ref{thmbb01}, are directly applicable to a much larger family of models, including high dimensional cases, models with jumps, and so on.

Let $0\leq t_0\leq T$. The \emph{realised variance} of $S_t$ during the time period $[t_0,T]$ is defined to be
\[
AF\sum_{i=1}^n \bigg(\log \frac{S_{t_i}}{S_{t_{i-1}}}\bigg)^2,
\]
where $t_0<t_1<\cdots<t_n=T$ are observation dates of $S_t$ and $AF$ is an annualisation factor. For example, if $t_i$ corresponds to daily observations then $AF=100^2\times 252/n$ and the realised variance is expressed in basis points per annum. As the mesh of the partition $\pi^n=\{t_0<t_1<\cdots<t_n\}$ tends to zero, the realised variance $R=R(t_0,T)$ can be represented as the quadratic variation of $\log S_t$, given by
\begin{gather}\label{eqrealisedvar}
R=R(t_0,T):=\lim_{n\to \infty} AF\sum_{t_i\in\pi^n} \bigg(\log \frac{S_{t_i}}{S_{t_{i-1}}}\bigg)^2 = \frac{100^2}{T-t_0} \int_{t_0}^T \sigma(t,S_t, V_t)^2\, dt.
\end{gather}
In the driftless case of $\mu(t,S_t)=0$, it is well-known that the expression in \eqref{eqrealisedvar} is equivalent to the value of a contingent claim with payoff $-2\log(S_T/S_{t_0})$, which can be further expressed as
\begin{align*}
\int_{t_0}^T \sigma(t,S_t, V_t)^2\, dt &=-2\E\left(\log\frac{S_T}{S_{t_0}}\,\middle|\,\Filt_{t_0}\right)\\
&=2\int_0^{S_{t_0}}\frac{\E((k-S_T)^+\,|\,\Filt_{t_0})}{k}\,dk+2\int_{S_{t_0}}^{\infty}\frac{\E((S_T-k)^+\,|\,\Filt_{t_0})}{k}\,dk.
\end{align*}
Hence the realised variance is actually observable from the prices of call and put options.
The \emph{VIX} $I=I(t_0,T)$ is defined to be the square root of the expected realised variance, 
\begin{gather*}
I(t_0,T)=\sqrt{\E(R(t_0,T)\,|\,\Filt_{t_0})} = 100 \times\sqrt{\frac{1}{T-t_0}\E\bigg(\int_{t_0}^T \sigma(t,S_t, V_t)^2\, dt\,\bigg|\,\Filt_{t_0}\bigg)}.
\end{gather*}
The VIX has a one month time horizon, or $T-t_0=1/12$.

Common derivatives on the VIX include futures, swaps, call options and put options. We will mostly focus on the the pricing of VIX futures and the VIX caps, which involves the computation of the following expectations:
\begin{align}\label{eqaa01}
u^f&:=\E (I(t_0,T)) = \E(\sqrt{\E(R(t_0,T)\,|\,\Filt_{t_0})}),\\
\label{eqaa02} u^c&:=\E (\min(I(t_0,T),K))= \E(\min(\sqrt{\E(R(t_0,T)\,|\,\Filt_{t_0})},K)).
\end{align}
Many other derivatives such as swaps, calls and puts can then be simply written in terms of $u^f$ and $u^c$:
\begin{align}\label{eqoptiondef1}
u^{swap}&:=\E (I(t_0,T)-K) = u^f-K,\\
u^{call}&:=\E(I(t_0,T)-K)^+ = u^f-u^c,\\
u^{put}&:=\E (K-I(t_0,T))^+ = K-u^c.\label{eqoptiondef3}
\end{align}
Note that if we were working in a model with stochastic interest rates, then forward prices will be used instead of futures in \eqref{eqoptiondef1}--\eqref{eqoptiondef3}.

\section{Upper and Lower Bounds}\label{sec3}

During the numerical pricing of VIX derivatives via Monte Carlo simulations, the main challenge is the computation of the inner conditional expectation in \eqref{eqaa01} and \eqref{eqaa02}, $\E(R(t_0,T)\,|\,\Filt_{t_0})$. This can be achieved by nested simulations or a least square Monte Carlo. In this section, we assume the exact value of $\E(R(t_0,T)\,|\,\Filt_{t_0})$ is unavailable, and propose a new Monte Carlo approach which produces true upper and lower bounds for VIX derivatives. 
This approach is similar to the well-known duality bounds for Bermudan and American options.

We will first briefly describe the duality bounds for a Bermudan or American option. For a more detailed exposition, the readers are referred to Rogers \cite{Rogers} or Haugh and Kogan \cite{Haugh}. Suppose that payoff process of the option is $Z$. The holder of the option chooses $\tau\in\STOP$ where $\STOP$ is the set of stopping times with values in $[0,T]$, corresponding to the available exercise opportunities (discrete in Bermuan, continuous in American). For any chosen $\tau$, the holder receives the payoff of $Z_\tau$ at time $\tau$. It is well-known that at time $t\in [0,T]$ the price of the option is given by $V_t=\esssup_{\tau\in\STOP} \E Z_\tau$, and that the price process $V$ is a supermartingale. It is clear that a lower bound of the option price $V_0$ can be found by selecting any sub-optimal stopping time $\tau'$ and computing $\E Z_{\tau'}$, and equality is achieved if $\tau'=\tau^*$ is the optimal stopping time. For an upper bound, let $M$ is an arbitrary martingale and consider $M_0+\E(\sup_t Z_t-M_t)$ where the supremum inside the expectation is taken path-wise. The validity of this upper bound can be checked by exchanging the expectation with the supremum and applying the optional sampling theorem. Equality is reached if the martingale $M$ is taken from the Doob-Meyer decomposition of the price process $V$, which can also be interpreted as the hedging portfolio. To summaries, bounds for the option price $V_0$ are given by
\[
 \E Z_\tau \leq V_0 \leq M_0+\E\left(\sup_{t\in[0,T]} Z_t-M_t\right),
\]
where $\tau$ is an arbitrary stopping time and $M$ is an arbitrary martingale.

A similar technique will be applied to obtain bounds for the VIX. These theoretical bounds rely on the following theorem.
\begin{theorem}\label{propav01}
Let $D\subseteq\R$ be an interval and $f:D\to\R$ be a concave function. Let $H$ be a $\Filt_T$-measurable random variable such that both $H$ and $f(H)$ are integrable. Fix $t_0\in[0,T]$.

(i) Suppose that $f^*:D^*\to\R$ is the concave conjugate of $f$, that is,
\begin{align}\label{eqau01}
f^*(y):=\inf_{x\in D}(xy-f(x)),
\end{align}
and $D^*$ is the domain of $f^*$ such that the infimum in \eqref{eqau01} is well defined. Then
\begin{align}\label{eqau02}
\E\Big( f(\E(H\,|\,\Filt_{t_0}))\Big)=\inf_{Y\in \mathcal{Y}_{t_0}}\E\Big(YH-f^*(Y)\Big),
\end{align}
where $\mathcal{Y}_{t_0}$ is the set of $\Filt_{t_0}$ measurable, integrable random variables taking values in $D^*$.

(ii)
We also have the equality
\begin{align}\label{eqav01}
\E\Big( f(\E(H\,|\,\Filt_{t_0}))\Big)=\sup_{M\in \mathcal{M}_{t_0}} \E\Big(f(H-M_{T})\Big),
\end{align}
where and $\mathcal{M}_{t_0}$ is the set of martingales which vanish at time $t_0$. Note that in \eqref{eqav01}, we have adopted the convention of $f(x)=-\infty$ for $x\notin D$.
\end{theorem}
\begin{proof}
(i) By the definition of the concave conjugate and the Legendre transform, $f^*$ satisfies
\[
f(x)=\inf_{y\in D^*}(xy-f^*(y)),
\]
and for each $x$, there exists a $y^*$ where equality is reached. Then \eqref{eqau01} follows by substituting $x$ with $\E(H\,|\,\Filt_{t_0})$ and noting that $y$ can be chosen according to the value of $\E(H\,|\,\Filt_{t_0})$, thus equality is attained for a $Y$ which is $F_{t_0}$-measurable.

(ii) By Jensen's inequality,
\[
\E\Big(f(H-M_{T})\Big) \leq \E\Big(f(\E(H-M_{T}\,|\,\Filt_{t_0}))\Big) = \E\Big( f(\E(H\,|\,\Filt_{t_0}))\Big).
\]
Furthermore, equality can be achieved by choosing the martingale defined by $M_t=\E(H-\E(H\,|\,\Filt_{t_0})\,|\,\Filt_t)$. Thus \eqref{eqav01} is established. Note that the result still holds if we relax the set $\mathcal{M}_{t_0}$ to include submartingales.
\end{proof}

Since $\sqrt{x}$ and $\min(\sqrt{x},c)$ are concave functions, Theorem \ref{propav01} provides natural bounds for VIX futures and caps. The quality of the bounds depends on the exact choice of $Y$ in \eqref{eqau02} and $M$ in \eqref{eqav01}. However, there is a problem with the lower bound
\begin{gather}\label{eqba02}
u^f\geq \E\left(\sqrt{R-M_T}\right),
\end{gather}
since for many choices of $M$, $M_T$ would exceed $R$ with non-zero probability, which then leads to the unusable lower bound of $-\infty$. This issue is resolved by the following theorem.

\begin{theorem}\label{thmbb01}
Denote the realised variance over $[t_0,T]$ by $R=R(t_0,T)$. Let $X$ be any positive $\Filt_{t_0}$-measurable random variable and $M$ be any martingale with $M_{t_0}=0$. Then we have the following inequalities.

(i) The VIX future price $u^f=\E( \sqrt{\E(R\,|\,\Filt_{t_0})})$ satisfies
\begin{gather}\label{eqbb01}
\E\left(\frac{R}{2\sqrt{X}}+\frac{\sqrt{X}}{2}\right) \geq u^f 
 \geq \E\left(\sqrt{(R-M_{T})^+}\right)-\sqrt{\E\left(\sqrt{\max(R, M_{T})}-\sqrt{R}\right)^2},
\end{gather}
where $x^+=\max(x,0)$. Equalities are achieved when $X=\E(R\,|\,\Filt_{t_0})$ and $M_T=R-\E(R\,|\,\Filt_{t_0})$.

(ii) Fix $K>0$, the VIX cap price $u^c=\E(\min(\sqrt{\E(R\,|\,\Filt_{t_0})},K))$ satisfies
\begin{gather}\label{eqbb02}
\E\left(\left(\frac{R}{2\sqrt{X}}+\frac{\sqrt{X}}{2}\right)\I(X\leq K^2)+K\I(X>K^2)\right) \geq u^c \\
 \geq \E\left(\min\left(\sqrt{(R-M_{T})^+},K\right)\right)-\sqrt{\E\left(\sqrt{\max(R, M_{T})}-\sqrt{R}\right)^2}.
\end{gather}
Equalities are again achieved when $X=\E(R\,|\,\Filt_{t_0})$ and $M_T=R-\E(R\,|\,\Filt_{t_0})$.
\end{theorem}
\begin{proof} (i)
The function $\sqrt{x}$ has the following Legendre transform,
\begin{align}\label{eqav02}
\sqrt{x}&=\inf_{y>0}\bigg(xy+\frac{1}{4y}\bigg),
\end{align}
where the infimum is achieved by $y^*=\frac{1}{2\sqrt{x}}$. Then by Proposition \ref{propav01} (i), for any positive $\Filt_{t_0}$-measurable random variable $Y$, we have
\[
\E\left( \sqrt{\E(R\,|\,\Filt_{t_0})}\right) \leq \E\left( R Y+\frac{1}{4Y}\right).
\]
The upper bound in \eqref{eqbb01} follows from the substitution $Y=\frac{1}{2\sqrt{X}}$.

For the lower bound, first note the identity
\begin{gather}
\E(\max(R, M_{T})\,|\,\Filt_{t_0})=\E((R-M_{T})^+ +M_T\,|\,\Filt_{t_0})=\E((R-M_{T})^+\,|\,\Filt_{t_0}).
\end{gather}
Now the required bound can be derived as follows,
\begin{align}\label{eqbc01}
\E(\sqrt{\E(R\,|\,\Filt_{t_0})})&\geq \E\bigg(\sqrt{\E(\max(R, M_{T})\,|\,\Filt_{t_0})}-\sqrt{\E((\sqrt{\max(R, M_{T})}-\sqrt{R})^2\,|\,\Filt_{t_0})}\bigg)\\
&= \E\bigg(\sqrt{\E((R-M_{T})^+\,|\,\Filt_{t_0})}-\sqrt{\E((\sqrt{\max(R, M_{T})}-\sqrt{R})^2\,|\,\Filt_{t_0})}\bigg)\\
\label{eqbc03} &\geq \E\left(\sqrt{(R-M_{T})^+}\right)-\sqrt{\E\left(\sqrt{\max(R, M_{T})}-\sqrt{R}\right)^2}.
\end{align}
The first inequality is due to the triangle inequality while the last inequality is due to Jensen's inequality. Note that we switched from $\max(R, M_{T})$ to $(R-M_{T})^+$ since the latter typically has lower variance for desirable choices of $M_T$ (i.e., for $M_T\approx R-\E(R\,|\,\Filt_{t_0})$), leading to a tighter Jensen's inequality. The equality cases can be easily checked via substitution.

(ii) The VIX cap case is similar to (i) with a few adjustments. The function $\min(\sqrt{x},K)$ has a Legendre transform given by,
\begin{align}\label{eqav03}
\min(\sqrt{x},K)&=\inf_{y\geq 0}\bigg(xy+\frac{1}{4y}\I(2yK\geq 1)+(K-K^2y)\I(2yK<1)\bigg)
\end{align}
where the infimum is achieved by $y^*=\frac{1}{2\sqrt{x}}\I(x\leq K^2)$.
Again applying Proposition \ref{propav01} (i), we have the upper bound
\[
\E\left( \min(\sqrt{\E(R\,|\,\Filt_{t_0})},K)\right) \leq \E\left( R Y+\frac{1}{4Y}\I(2YK\geq 1)+(K-K^2Y)\I(2YK< 1)\right).
\]
This simplifies to the required upper bound in \eqref{eqbb02} after substituting $Y=\frac{1}{2\sqrt{X}}\I(X\leq K^2)$.

The lower bound can be established by using the same argument as \eqref{eqbc01}--\eqref{eqbc03} in (i), combined with the inequality
\begin{align*}
&\sqrt{\E(\max(R, M_{T})\,|\,\Filt_{t_0})}-\sqrt{\E(R\,|\,\Filt_{t_0})} \\
&\qquad \geq  \min(\sqrt{\E(\max(R, M_{T})\,|\,\Filt_{t_0})},K)-\min(\sqrt{\E(R\,|\,\Filt_{t_0})},K).
\end{align*}
Note that we have used the fact that $\max(R, M_{T})\geq R$. Finally, the equality conditions can be checked by substitution.
\end{proof}
A key feature of the upper and lower bounds presented in Theorem \ref{thmbb01} is that they can all be computed using a standard Monte Carlo simulation.
The lower bound in \eqref{eqbb01} can be computed even if $\Prob(M_t>R)>0$. In the case where $R\geq M_t$ holds almost surely, it reduces to the simpler bound in \eqref{eqba02}, $\E(\sqrt{R-M_T})$.

\begin{remark}\label{remjensenbounds}
As an immediate consequence of Jensen's inequality, the value of the VIX future is bounded between the volatility swap and the square root of the variance swap, both evaluated at time 0,
\[
\E(\sqrt{R})\leq u^f \leq \sqrt{\E R}.
\]
Both of these bounds can be seen as special cases of Theorem \ref{thmbb01} (i), by setting $X$ to the variance swap $\sqrt{\E R}$ evaluated at time 0 and by setting $M_T$ to zero. Also, it is noteworthy that equality is reached in Theorem \ref{thmbb01} when $X$ is the variance swap evaluated at time $t_0$ and $M$ is hedging portfolio of the same variance swap during $[t_0,T]$. In practical implementations, if $M_T$ is poorly estimated and $M_t>R$ occurs frequently, it may be more advantageous to simply use $\E(\sqrt{R})$ as a lower bound instead.
\end{remark}

\begin{remark}\label{remchooser}
The upper bound in Theorem \ref{thmbb01} has the following interesting interpretation. The VIX future can be represented as a variant of the chooser option on the realised variance. In particular, the seller of the option may select a non-negative real $x$ at time $t_0$, and then must pay the holder $R/(2\sqrt{x}) + \sqrt{x}/2$ at time $T$. If the seller chooses optimally, i.e., minimising the expected payoff at time $t_0$, then the value of the option coincides with the VIX future.
\end{remark}

\section{Least Squares Monte Carlo}\label{sec4}

In this section, we shall described the empirical Monte Carlo algorithm used to compute bounds for VIX derivatives. The algorithm utilises a variant of the least squares Monte Carlo proposed by Schoenmakers et al.\! \cite{Schoenmakers} which simultaneously estimates the conditional expectation as well as the martingale increment. We refer the readers to Schoenmakers et al.\! \cite{Schoenmakers} for results regarding stability and convergence of the method, as well as Joshi and Tang \cite{Joshi} for an overview of related methods.

Suppose that the time interval $[t_0,T]$ is partitioned into $t_0<t_1<\cdots<t_n=T$. First simulate $N$ trajectories $S^i$ and $V^i$ for $i=1,\ldots,N$, and compute the corresponding realised variances $R^i$. Recall that, by Theorem \ref{thmbb01}, in order to obtain good quality bounds on VIX derivatives, it is important to find good approximations to the conditional expectation $X=\E(R\,|\,\Filt_{t_0})$ and the martingale increment $M_T=R-\E(R\,|\,\Filt_{t_0})$. We postulate that $X$ and $M_T$ can be approximated in terms of the state variables in the following way:
\begin{gather}
X=\E(R\,|\,\Filt_{t_{0}})   \approx \Psi(S_{t_0},V_{t_0}) := \sum_{j=1}^p \beta_j \psi_j(S_{t_0},V_{t_0}),\\
M_T=R-\E(R\,|\,\Filt_{t_0})=\sum_{l=0}^{n-1}\E(R\,|\,\Filt_{t_{l+1}})-\E(R\,|\,\Filt_{t_l})\approx \sum_{l=0}^{n-1}\Phi_{t_l}(S_{t_l},V_{t_l})\cdot\Delta W_{t_l},\\ 
\Psi(s,v) := \sum_{j=1}^p \beta_j \psi_j(s,v),\quad \Phi_{t_l}(s,v):=\sum_{j=1}^q \gamma_{j,{l}} \phi_j(s,v),
\end{gather}
where $\psi_j : \R^2\to\R$ and $\phi_j : \R^2\to \R^2$ are appropriate basis functions chosen beforehand. Note that $\Delta W_{t_l}:=(W^S_{t_{l+1}}-W^S_{t_{l}},W^V_{t_{l+1}}-W^V_{t_{l}})'$. In practice $\Delta W_{t_l}$ can be replaced by other appropriate martingale increments with the predictable representation property.

\begin{remark}
Due to the Markov properties of the model and the predictable representation theorem, if the space spanned by the basis function is rich enough, the conditional expectation can be matched exactly while the martingale increment will be replicated as the mesh of the partition goes to 0,
\[
\E(R\,|\,\Filt_{t_{0}})   = \Psi(S_{t_0},V_{t_0}),\quad R-\E(R\,|\,\Filt_{t_0})=\int_{t_0}^T \Phi_{t}(S_{t},V_{t})\cdot dW_{t}.
\]
\end{remark}

The coefficients 
\[
B=(\beta_j: j=1,\ldots, p),\quad \Gamma=(\gamma_{j,l}:j=1,\ldots, q; l=1,\ldots, n)
\]
are estimated in the linear least squares regression problem:
\begin{align*}
(\hat B,\hat \Gamma)&=\argmin_{B\in \R^p, \Gamma\in \R^{q\times n}} \sum_{i=1}^N\bigg(R^i - \Psi(S^i_{t_0},V^i_{t_0}) - \sum_{l=0}^{n-1}\Phi_{t_l}(S^i_{t_l},V^i_{t_l})\cdot\Delta W^i_{t_l}\bigg)^2\\
&=\argmin_{B\in \R^p, \Gamma\in \R^{q\times n}} \sum_{i=1}^N\bigg(R^i - \sum_{j=1}^p \beta_j \psi_j(S^i_{t_0},V^i_{t_0})- \sum_{l=0}^{n-1}\sum_{j=1}^q \gamma_{j,{l}} \phi_j(S^i_{t_l},V^i_{t_l})\cdot\Delta W^i_{t_l}\bigg)^2.
\end{align*}
Let us the denote the estimated functions by
\[
\hat\Psi(s,v) = \sum_{j=1}^p \hat\beta_j \psi_j(s,v),\quad \hat\Phi_{t_l}(s,v)=\sum_{j=1}^q \hat\gamma_{j,{l}} \phi_j(s,v).
\]

In order to compute true upper and lower bounds, we generate a new set of $\tilde N$ trajectories $\tilde S^i$ and $\tilde V^i$ for $i=1,\ldots,\tilde N$. 
This is performed to avoid the foresight bias caused by reusing the original trajectories. A detailed explanation of the foresight bias can be found in Fries \cite{Fries}.
Our new path-wise estimates of the conditional expectation and the martingale increment are
\[
\hat X^i=\hat\Psi(\tilde S^i_{t_0},\tilde V^i_{t_0}), \quad \hat M^i_T=\sum_{l=0}^{n-1}\hat\Phi_{t_l}(\tilde S_{t_l},\tilde V_{t_l})\cdot\Delta \tilde W^i_{t_l}.
\]
At this point we apply Theorem \ref{thmbb01} on the estimates $\hat X^i$ and $\hat M^i_T$ to produce bounds for the VIX future and cap. Specifically, we have
\begin{align}
\overline u^f&=\frac{1}{\tilde N}\sum_{i=1}^{\tilde N}\left(\frac{ \tilde R^i}{2\sqrt{\hat X^i}}+\frac{\sqrt{\hat X^i}}{2}\right),\label{eqccf1}\\
\underline u^f&= \frac{1}{\tilde N}\sum_{i=1}^{\tilde N}\left( \sqrt{(\tilde R^i-\hat M^i_{T})^+}\right)-\sqrt{\frac{1}{\tilde N}\sum_{i=1}^{\tilde N}\left( \sqrt{\max(\tilde R^i, \hat M^i_{T})}-\sqrt{\tilde R^i}\right)^2},\label{eqccf2}\\
\overline u^c&=\frac{1}{\tilde N}\sum_{i=1}^{\tilde N}\left(\left(\frac{ \tilde R^i}{2\sqrt{\hat X^i}}+\frac{\sqrt{\hat X^i}}{2}\right)\I(\hat X^i\leq K^2)+ K\I(\hat X^i>K^2)\right),\\
\underline u^c&= \frac{1}{\tilde N}\sum_{i=1}^{\tilde N}\left( \min\left(\sqrt{(\tilde R^i-\hat M^i_{T})^+},K\right)\right)-\sqrt{\frac{1}{\tilde N}\sum_{i=1}^{\tilde N}\left( \sqrt{\max(\tilde R^i, \hat M^i_{T})}-\sqrt{\tilde R^i}\right)^2}.
\end{align}
Note that the realised variances $\tilde R^i$ are directly computed from $\tilde S^i$ and $\tilde V^i$. 
Bounds for other derivatives such as swaps, calls and puts can now be easily computed:
\begin{alignat}{2}
\overline u^{swap}& = \overline u^f-K, &\quad \underline u^{swap}& = \underline u^f-K,\\
\overline u^{call}& = \overline u^f-\underline u^c,&\quad \underline u^{call}& = \underline u^f-\overline u^c,\\
\overline u^{put}&= K-\underline u^c, &\quad \underline u^{put}&= K-\overline u^c.
\end{alignat}

\begin{remark}
At a first glance, the term $\sqrt{\hat X^i}$ in the upper bound calculation could cause problems since $\hat X^i$ may be negative. In practical implementations, a floor is often imposed on the instantaneous volatility. It is then natural to enforce the same floor on $\hat X^i$,
\[
\hat X^i=\max(\hat\Psi(\tilde S^i_{t_0},\tilde V^i_{t_0}), h).
\]
The result will still be a true upper bound.
This is in contrast to the lower bound term $\sqrt{R-M_T}$ where the sign of $R-M_T$ is harder to control. A simple floor on $R-M_T$ will violate the validity of the lower bound. Thus Theorem \ref{thmbb01} was necessary to overcome this issue. In general, these issues can also be alleviated fit by using more and better basis functions, thus improving the least squares fit.
\end{remark}

\section{Numerical Results}\label{sec5}

For our numerical example, we choose the following variant of the CEV-Heston LSV model with volatility caps and floors:
\begin{align*}
dS_t&= \sigma(S_t,V_t) S_t dW^S_t,\\
dV_t&=\kappa(\theta-V_t) dt+ \eta\sqrt{V_t} dW^V_t,\\
\sigma(S_t,V_t) &= f(\sqrt{V_t} (S_t/S_0)^{\alpha-1}),\\
f(x)&=\max(\min(x,10),0.01),\\
\langle dW^S_t,dW^V_t\rangle & = \rho dt.
\end{align*}
This is essentially the same as the usual CEV-Heston model, but the effective volatility is bounded between $0.01$ and $10$.
Recall that the interest rate is assumed to be zero. Table \ref{tabparam} contains our chosen parameter value as well as their interpretations.
\begin{table}[h]\centering
\begin{tabular}{|c|c|c|}
\hline
Parameter & Value & Interpretation \\
\hline
$S_0$ & 100 & initial stock price\\
$\alpha$ & 0.8 & leverage between stock and volatility\\
$\sigma(S_0, V_0)$ & 0.3 & initial volatility\\
$V_0$ & 0.09 & initial variance\\
$\kappa$ & 0.6 & mean-reversion speed\\
$\theta$ & 0.09 & long term variance\\
$\eta$ & 0.4 & vol of vol\\
$\rho$ & -0.5 & correlation between stock and variance\\
$t_0$ & 1 & VIX start date\\
$T$ & 1+1/12 & VIX end date\\
$\Delta t$ & 1/120 & time increment\\
$N$ & 100000 & paths for regression\\
$\tilde N$ & 500000 & paths for bound calculation\\
\hline
\end{tabular}
\caption{Parameter values and interpretations}\label{tabparam}
\end{table}

We will be employing the algorithm described in Section \ref{sec4} to compute bounds for VIX futures, caps, calls and puts. The simulation scheme used will be the standard Euler scheme with full truncation. Antithetic variables are used for variance reduction. During the regression step, the following basis functions are used:
\begin{align*}
\Psi(s,v) &:= \sum_{j=1}^p \beta_j \psi_j(\log s,\sqrt{v}),\\
\Phi_{t_l}(s,v)&:=\sum_{j=1}^p \gamma_{j,{l}} \phi_j(\log s,\sqrt{v}) \left(\sigma(s,v) s \frac{d}{ds}\log s, \eta\sqrt{v}\frac{d}{dv}\sqrt{v} \right)'\\
&=\sum_{j=1}^p \gamma_{j,{l}} \phi_j(\log s,\sqrt{v}) \left(\sigma(s,v), \frac{\eta}{2}\right)',
\end{align*}
where $\psi_j$ and $\phi_j$ are bivariate polynomials. Two cases are examined: lower degree polynomials where $\psi_j$ and $\phi_j$ have degrees 3 and 2 respectively, and higher degree polynomials where $\psi_j$ and $\phi_j$ have degrees 4 and 3 respectively. During the computation of upper bounds, the volatility cap and floor function (i.e., $f$) is also applied to $\hat X$. In the computation of lower bounds, the martingale increments can be interpreted as the delta and vega hedging strategies.

\begin{table}[hp]\centering
\begin{tabular}{|c|c|c|c|c|}
\hline
& $u^f$ & $\hat u^f$ & $\underline u^f$ & $\overline u^f$ \\
\hline
& 27.3728   $\pm$   0.0445  &  27.3500   $\pm$   0.0446  &  27.3582   $\pm$   0.0445  &  27.4607   $\pm$   0.0454  \\
\hline
&  &   & $\E\sqrt{R}=27.1018$ & $\sqrt{\E R}=31.7342$\\
\hline
\hline
$K$ & $u^{call}$ & $\hat u^{call}$ & $\underline u^{call}$ & $\overline u^{call}$ \\
\hline
15  &  13.7302  $\pm$  0.0404  &  13.7417  $\pm$  0.0403  &  13.6269  $\pm$  0.0676  &  13.8326  $\pm$  0.0415  \\
20  &  10.2909  $\pm$  0.0367  &  10.2887  $\pm$  0.0367  &  10.1892  $\pm$  0.0817  &  10.3932  $\pm$  0.0380  \\
25  &  7.4785  $\pm$  0.0324  &  7.4672  $\pm$  0.0324  &  7.3775  $\pm$  0.0939  &  7.5809  $\pm$  0.0339  \\
30  &  5.2738  $\pm$  0.0280  &  5.2600  $\pm$  0.0280  &  5.1716  $\pm$  0.1021  &  5.3763  $\pm$  0.0297  \\
35  &  3.6176  $\pm$  0.0236  &  3.6065  $\pm$  0.0236  &  3.5156  $\pm$  0.1052  &  3.7202  $\pm$  0.0257  \\
40  &  2.4230  $\pm$  0.0196  &  2.4169  $\pm$  0.0196  &  2.3208  $\pm$  0.1034  &  2.5256  $\pm$  0.0221  \\
45  &  1.5912  $\pm$  0.0160  &  1.5900  $\pm$  0.0161  &  1.4887  $\pm$  0.0978  &  1.6938  $\pm$  0.0191  \\
\hline
\hline
$K$ & $u^{put}$ & $\hat u^{put}$ & $\underline u^{put}$ & $\overline u^{put}$ \\
\hline
15  &  1.3575  $\pm$  0.0079  &  1.3916  $\pm$  0.0083  &  1.2686  $\pm$  0.0091  &  1.3719  $\pm$  0.0079  \\
20  &  2.9181  $\pm$  0.0131  &  2.9386  $\pm$  0.0134  &  2.8310  $\pm$  0.0141  &  2.9326  $\pm$  0.0131  \\
25  &  5.1057  $\pm$  0.0185  &  5.1172  $\pm$  0.0187  &  5.0193  $\pm$  0.0193  &  5.1203  $\pm$  0.0185  \\
30  &  7.9010  $\pm$  0.0236  &  7.9100  $\pm$  0.0238  &  7.8134  $\pm$  0.0245  &  7.9157  $\pm$  0.0236  \\
35  &  11.2449  $\pm$  0.0282  &  11.2565  $\pm$  0.0284  &  11.1574  $\pm$  0.0291  &  11.2595  $\pm$  0.0283  \\
40  &  15.0502  $\pm$  0.0322  &  15.0669  $\pm$  0.0323  &  14.9625  $\pm$  0.0331  &  15.0650  $\pm$  0.0322  \\
45  &  19.2184  $\pm$  0.0354  &  19.2400  $\pm$  0.0355  &  19.1305  $\pm$  0.0363  &  19.2331  $\pm$  0.0354  \\
\hline
\end{tabular}
\caption{Lower degree polynomials results for VIX futures, calls and puts, including nested Monte Carlo results, least square Monte Carlo estimates, as well as lower and upper bounds}\label{tabmain}
\end{table}
\begin{table}[hp]\centering
\begin{tabular}{|c|c|c|c|c|}
\hline
& $u^f$ & $\hat u^f$ & $\underline u^f$ & $\overline u^f$ \\
\hline
& 27.3728   $\pm$   0.0445  &  27.3739   $\pm$   0.0445  &  27.3707   $\pm$   0.0445  &  27.3751   $\pm$   0.0455    \\
\hline
&  &   & $\E\sqrt{R}=27.1018$ & $\sqrt{\E R}=31.7342$\\
\hline
\hline
$K$ & $u^{call}$ & $\hat u^{call}$ & $\underline u^{call}$ & $\overline u^{call}$ \\
\hline
15  &  13.7302  $\pm$  0.0404  &  13.7313  $\pm$  0.0404  &  13.7265  $\pm$  0.0673  &  13.7346  $\pm$  0.0415  \\
20  &  10.2909  $\pm$  0.0367  &  10.2921  $\pm$  0.0367  &  10.2883  $\pm$  0.0814  &  10.2952  $\pm$  0.0380  \\
25  &  7.4785  $\pm$  0.0324  &  7.4793  $\pm$  0.0324  &  7.4756  $\pm$  0.0936  &  7.4828  $\pm$  0.0338  \\
30  &  5.2738  $\pm$  0.0280  &  5.2741  $\pm$  0.0280  &  5.2701  $\pm$  0.1018  &  5.2781  $\pm$  0.0296  \\
35  &  3.6176  $\pm$  0.0236  &  3.6173  $\pm$  0.0236  &  3.6140  $\pm$  0.1050  &  3.6220  $\pm$  0.0255  \\
40  &  2.4230  $\pm$  0.0196  &  2.4224  $\pm$  0.0196  &  2.4191  $\pm$  0.1034  &  2.4274  $\pm$  0.0218  \\
45  &  1.5912  $\pm$  0.0160  &  1.5904  $\pm$  0.0160  &  1.5871  $\pm$  0.0978  &  1.5955  $\pm$  0.0187  \\
\hline
\hline
$K$ & $u^{put}$ & $\hat u^{put}$ & $\underline u^{put}$ & $\overline u^{put}$ \\
\hline
15  &  1.3575  $\pm$  0.0079  &  1.3574  $\pm$  0.0079  &  1.3558  $\pm$  0.0088  &  1.3595  $\pm$  0.0079  \\
20  &  2.9181  $\pm$  0.0131  &  2.9182  $\pm$  0.0131  &  2.9176  $\pm$  0.0140  &  2.9202  $\pm$  0.0131  \\
25  &  5.1057  $\pm$  0.0185  &  5.1054  $\pm$  0.0185  &  5.1049  $\pm$  0.0194  &  5.1077  $\pm$  0.0185  \\
30  &  7.9010  $\pm$  0.0236  &  7.9001  $\pm$  0.0236  &  7.8994  $\pm$  0.0246  &  7.9031  $\pm$  0.0236  \\
35  &  11.2449  $\pm$  0.0282  &  11.2434  $\pm$  0.0282  &  11.2432  $\pm$  0.0292  &  11.2469  $\pm$  0.0282  \\
40  &  15.0502  $\pm$  0.0322  &  15.0484  $\pm$  0.0322  &  15.0484  $\pm$  0.0332  &  15.0523  $\pm$  0.0322  \\
45  &  19.2184  $\pm$  0.0354  &  19.2165  $\pm$  0.0354  &  19.2164  $\pm$  0.0364  &  19.2205  $\pm$  0.0354  \\
\hline
\end{tabular}
\caption{Higher degree polynomials results for VIX futures, calls and puts, including nested Monte Carlo results, least square Monte Carlo estimates, as well as lower and upper bounds}\label{tabmain2}
\end{table}

As an analytical benchmark, we will be using the results of a nested Monte Carlo. In this simulation, 500000 trajectories are generated up to time $t_0$. Then for each of these trajectories, a sub-simulation of 5000 trajectories is carried out on the time interval $[t_0,T]$ to compute the conditional expectation $\E(R\,|\,\Filt_{t_0})$ path-wise. The prices of the VIX derivatives are computed by averaging the relevant payoffs over all trajectories. In order to check the correctness of our bounds, the same 500000 paths on $[0,t_0]$ from the nested Monte Carlo will also be used in the second simulation of our least squares Monte Carlo. After that, the behaviour of the paths on $[t_0,T]$ are generated independently for the different methods. This allows us to compare the relative sizes of the results without the effects of variances due to simulation. In terms of computation times, the least squares Monte Carlo method is more than 1000 times faster than the nested method.

The results for lower degree polynomials are found in Tables \ref{tabmain} while the higher degree polynomials results are found in \ref{tabmain2}. In terms of notations, for VIX futures: $u^f$ is the analytical value computed using the nested Monte Carlo; $\hat u^f$ is the result of the classic least squares Monte Carlo by simply averaging the square root of the regression fit $\hat\Psi(\tilde S^i_{t_0},\tilde V^i_{t_0})$; $\underline u^f$ and $\overline u^f$ are the lower and upper bounds computed as described in \eqref{eqccf1} and \eqref{eqccf2}. For completeness, we have also included estimates for the volatility swap $\E\sqrt{R}$ and the square root of the variance swap $\sqrt{\E R}$. Similar notations are used for calls and puts over a range of strikes $K$. All confidence intervals are computed as 1.96 times the standard deviation. All values have also been annualised accordingly.

\begin{figure}[th]
\centering
\includegraphics[scale=0.7]{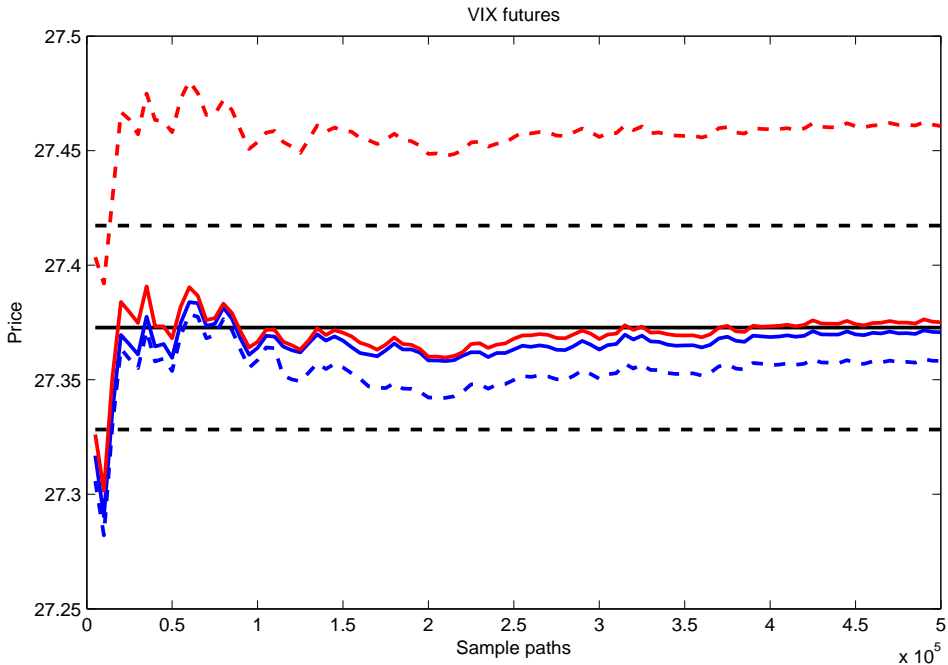}
\caption{Plot of VIX future bounds for different number of simulation paths. The solid black line is the nest Monte Carlo result, with its confidence interval indicated by the dashed black lines. The solid red and blue lines are the upper and lower bounds using higher degree polynomials. The dashed red and blue lines are the upper and lower bounds using lower degree polynomials.}\label{fig1}
\end{figure}
As shown in Table \ref{tabmain}, even with lower degree polynomials, our method produces tight bounds across all VIX derivatives and at all strike levels. In many cases the classical least squares Monte Carlo estimates actually fall outside of our bounds. Our bounds are also clearly superior when compared to the bounds given by the volatility and variance swaps (see Remark \ref{remjensenbounds}). In the higher degree polynomials case shown in Table \ref{tabmain2}, the convergence of our method is verified by the fact that all four estimates are extremely close. In fact, the difference between the estimates is much smaller than the corresponding confidence intervals. This indicates that, in terms of the bias-variance trade-off, most of the error comes from the variance caused by the number of simulation paths, while our method with higher degree polynomials has very little bias due to an excellent regression fit. Figure \ref{fig1} further illustrates this by plotting the VIX future bounds over different number of simulation paths. As the number of paths increases, the bounds stablise towards their limits. The higher degree polynomials results are noticeably better than the low degree results, especially in the upper bound which benefited greatly from the degree of $\psi_j$ increasing from 3 to 4.

Now we examine the effect of varying a single parameter on VIX futures. The other parameters are kept as in Table \ref{tabparam} and lower degree polynomials are used. First of all, we vary the correlation coefficient $\rho$. As shown in Table \ref{tabrho}, larger values of $\rho$ lead to lower VIX future prices. Our method works very well in all cases, especially for higher correlations. This is due to the fact that a negative correlation combined with a leverage coefficient satisfying $\alpha<1$ will lead to larger variations in the realised variance.
\begin{table}[h]\centering
\begin{tabular}{|c|c|c|c|}
\hline
$\rho$ & $\hat u^{f}$ & $\underline u^{f}$ & $\overline u^{f}$ \\
\hline
-0.8  &  27.6457  $\pm$  0.0471  &  27.6439  $\pm$  0.0469  &  27.7969  $\pm$  0.0478 \\
-0.6  &  27.4260  $\pm$  0.0453  &  27.4411  $\pm$  0.0452  &  27.5687  $\pm$  0.0461 \\
-0.4  &  27.2558  $\pm$  0.0439  &  27.2629  $\pm$  0.0439  &  27.3397  $\pm$  0.0448 \\
-0.2  &  27.0511  $\pm$  0.0423  &  27.0549  $\pm$  0.0423  &  27.0924  $\pm$  0.0432 \\
0.0  &  26.8643  $\pm$  0.0409  &  26.8646  $\pm$  0.0408  &  26.8888  $\pm$  0.0418 \\
0.2  &  26.6778  $\pm$  0.0395  &  26.6776  $\pm$  0.0395  &  26.6911  $\pm$  0.0404 \\
0.4  &  26.4776  $\pm$  0.0381  &  26.4784  $\pm$  0.0381  &  26.4867  $\pm$  0.0390 \\
0.6  &  26.2928  $\pm$  0.0368  &  26.2955  $\pm$  0.0368  &  26.3014  $\pm$  0.0377 \\
0.8  &  26.1158  $\pm$  0.0355  &  26.1201  $\pm$  0.0355  &  26.1235  $\pm$  0.0365 \\
\hline
\end{tabular}
\caption{VIX futures for different correlation values}\label{tabrho}
\end{table}

Next, we vary the vol of vol $\eta$ in Table \ref{tabvolvol}. As $\eta$ increases the VIX future decreases. For small values of $\eta$, the upper and lower bounds are essentially the same value. For extremely large values of $\eta$, the quality of the lower bound deteriorates substantially. 
\begin{table}[h]\centering
\begin{tabular}{|c|c|c|c|}
\hline
$\eta$ & $\hat u^{f}$ & $\underline u^{f}$ & $\overline u^{f}$ \\
\hline
0.1  &  30.2498  $\pm$  0.0140  &  30.2497  $\pm$  0.0140  &  30.2500  $\pm$  0.0143 \\
0.2  &  29.6859  $\pm$  0.0245  &  29.6864  $\pm$  0.0245  &  29.6864  $\pm$  0.0250 \\
0.3  &  28.6812  $\pm$  0.0349  &  28.6876  $\pm$  0.0349  &  28.6972  $\pm$  0.0356 \\
0.4  &  27.3309  $\pm$  0.0445  &  27.3381  $\pm$  0.0444  &  27.4319  $\pm$  0.0453 \\
0.5  &  25.9168  $\pm$  0.0531  &  25.8331  $\pm$  0.0530  &  26.2101  $\pm$  0.0541 \\
0.6  &  24.5632  $\pm$  0.0605  &  24.2752  $\pm$  0.0607  &  25.0365  $\pm$  0.0621 \\
0.7  &  23.5007  $\pm$  0.0668  &  22.5781  $\pm$  0.0676  &  24.4188  $\pm$  0.0719 \\
0.8  &  22.5599  $\pm$  0.0726  &  21.2048  $\pm$  0.0742  &  23.5922  $\pm$  0.0806 \\
\hline
\end{tabular}
\caption{VIX futures for different vol of vol}\label{tabvolvol}
\end{table}

Finally, Table \ref{tabalpha} examines the effect of varying the leverage coefficient $\alpha$. The bounds deteriorate somewhat for small values of $\alpha$. This is due to the negative correlation $\rho$, which creates more extreme values of the realised variance for small values of $\alpha$. The reverse would be true if $\rho$ was positive.
\begin{table}[h]\centering
\begin{tabular}{|c|c|c|c|}
\hline
$\alpha$ & $\hat u^{f}$ & $\underline u^{f}$ & $\overline u^{f}$ \\
\hline
0.7  &  27.9683  $\pm$  0.0492  &  26.5554  $\pm$  0.0492  &  28.5024  $\pm$  0.0506 \\
0.8  &  27.3445  $\pm$  0.0445  &  27.3564  $\pm$  0.0445  &  27.4246  $\pm$  0.0453 \\
0.9  &  26.8659  $\pm$  0.0414  &  26.8684  $\pm$  0.0414  &  26.8716  $\pm$  0.0424 \\
1.0  &  26.4738  $\pm$  0.0392  &  26.4725  $\pm$  0.0392  &  26.4735  $\pm$  0.0402 \\
1.1  &  26.1138  $\pm$  0.0373  &  26.1141  $\pm$  0.0373  &  26.1145  $\pm$  0.0383 \\
1.2  &  25.8070  $\pm$  0.0360  &  25.8049  $\pm$  0.0360  &  25.8162  $\pm$  0.0370 \\
1.3  &  25.5235  $\pm$  0.0350  &  25.5291  $\pm$  0.0349  &  25.5564  $\pm$  0.0358 \\
1.4  &  25.2905  $\pm$  0.0342  &  25.2924  $\pm$  0.0341  &  25.3685  $\pm$  0.0350 \\
1.5  &  25.0769  $\pm$  0.0335  &  25.0721  $\pm$  0.0335  &  25.1879  $\pm$  0.0344 \\
\hline
\end{tabular}
\caption{VIX futures for different leverage coefficients}\label{tabalpha}
\end{table}

Even though lower degree polynomials are used in Tables \ref{tabrho}, \ref{tabvolvol} and \ref{tabalpha}, our method generally works very well. In fact, in many cases the bounds are even better than the direct estimates $\hat u^{f}$ obtained from the classical least squares regression approach. Since the tightness of our bounds depends on the quality of the regression fit, the method understandably performs worse when there are extreme variations in the realised variance. This is particularly noticeable for the lower bound as a poor regression fit often leads to frequent occurrences of $M_T>R$. 
In these extreme cases, the results can be improved by using better basis functions.
Alternatively, one may also use the volatility swap $\E\sqrt{R}$ as a replacement lower bound.

\section{Conclusion}\label{sec6}
We have introduced a new model independent technique for the computation of true upper and lower bounds for VIX derivatives. Theorem \ref{propav01} includes a general stochastic duality result on payoffs involving concave functions. This is then applied to VIX derivatives in Theorem \ref{thmbb01}, along with minor adjustments to handle issues caused by the square root function. The upper bound involves the evaluation of a variance swap, while the lower bound involves estimating a martingale increment corresponding to its hedging portfolio. Our bounding technique is particularly useful in complex models where it is difficult to directly compute VIX derivative prices.
Numerically, a single linear least squares Monte Carlo method is used to simultaneously compute the upper and lower bounds. The method is shown to work very well for VIX futures, calls and puts under a wide range of parameter choices.


\begin{thebibliography}{99}

\bibitem{Andersen} Andersen, L., \& Broadie, M. (2004). Primal-dual simulation algorithm for pricing multidimensional American options. \emph{Management Science}, 50(9), 1222-1234.

\bibitem{Baldeaux} Baldeaux, J., \& Badran, A. (2014). Consistent modelling of VIX and equity derivatives using a 3/2 plus jumps model. \emph{Applied Mathematical Finance}, 21(4), 299--312.

\bibitem{Cont} Cont, R., \& Kokholm, T. (2013). A consistent pricing model for index options and volatility derivatives. \emph{Mathematical Finance}, 23(2), 248--274.

\bibitem{Detemple} Detemple, J., \& Osakwe, C. (2000). The valuation of volatility options. \emph{European Finance Review}, 4(1), 21--50.

\bibitem{Fries} Fries, C. P. (2008). Foresight Bias and Suboptimality Correction in Monte—Carlo Pricing of Options with Early Exercise. \emph{In Progress in Industrial Mathematics at ECMI 2006} (pp. 645--649). Springer Berlin Heidelberg.

\bibitem{Grunbichler} Gr\"unbichler, A., \& Longstaff, F. A. (1996). Valuing futures and options on volatility. \emph{Journal of Banking \& Finance}, 20(6), 985--1001.

\bibitem{Haugh} Haugh, M. B., \& Kogan, L. (2004). Pricing American options: a duality approach. \emph{Operations Research}, 52(2), 258-270.

\bibitem{Joshi} Joshi, M., \& Tang, R. (2014). Effective sub-simulation-free upper bounds for the Monte Carlo pricing of callable derivatives and various improvements to existing methodologies. \emph{Journal of Economic Dynamics and Control}, 40, 25--45.

\bibitem{Lian} Lian, G. H., \& Zhu, S. P. (2013). Pricing VIX options with stochastic volatility and random jumps. \emph{Decisions in Economics and Finance}, 36(1), 71--88.

\bibitem{Longstaff} Longstaff, F. A., \& Schwartz, E. S. (2001). Valuing American options by simulation: a simple least-squares approach. \emph{Review of Financial studies}, 14(1), 113-147.

\bibitem{Rogers} Rogers, L. C. (2002). Monte Carlo valuation of American options. \emph{Mathematical Finance}, 12(3), 271-286.

\bibitem{Schoenmakers} Schoenmakers, J., Zhang, J., \& Huang, J. (2013). Optimal dual martingales, their analysis, and application to new algorithms for Bermudan products. \emph{SIAM Journal on Financial Mathematics}, 4(1), 86--116.

\bibitem{Sepp} Sepp, A. (2008). VIX option pricing in a jump-diffusion model. \emph{Risk magazine}, 84--89.

\bibitem{Whaley} Whaley, R. E. (1993). Derivatives on market volatility: Hedging tools long overdue. \emph{The Journal of Derivatives}, 1(1), 71--84.

\bibitem{Zhang} Zhang, J. E., \& Zhu, Y. (2006). VIX futures. \emph{Journal of Futures Markets}, 26(6), 521--531.


\end{thebibliography}
\end{document}